\documentclass[journal]{IEEEtran}

\usepackage[textwidth=45pt, textsize=scriptsize, backgroundcolor=yellow!20!white,
  linecolor=blue]{todonotes}
\def\nn{\nonumber}

\usepackage{stfloats}

\usepackage{lipsum}
\usepackage{tikz}
\usetikzlibrary{patterns}

\usepackage{cancel}

\usepackage{amssymb, amsmath, amsthm, amsfonts} 

\usepackage{mathtools}

\usepackage{graphicx}
\ifCLASSOPTIONcompsoc
    \usepackage[caption=false, font=normalsize, labelfont=sf, textfont=sf]{subfig}
\else
\usepackage[caption=false, font=footnotesize]{subfig}
\fi

\usepackage[nice]{nicefrac}
\newcommand{\nf}[2]{\nicefrac{#1}{#2}}

\newcommand{\fig}[1]{Fig.~\ref{fig:#1}}

\newcommand{\eq}[1]{(\ref{eq:#1})}

\newcommand{\by}[1]{\! \times \!}

\def\ra{\!\rightarrow\!}

\newcommand{\mc}[1]{\mathcal{#1}}
\newcommand{\ob}[1]{\bar{#1}}
\newcommand{\sm}{\!\setminus\!}
\newcommand{\overbar}[1]{\ob{#1}}

\begin{document}

\title{Probability Mass Exclusions and the Directed Components of Pointwise Mutual Information}


\author{Conor~Finn 
  and~Joseph~T.~Lizier
  \thanks{C.~Finn (email:~conor.finn@sydney.edu.au),
    and J. T. Lizier (email: joseph.lizier@sydney.edu.au)
    are with the Complex Systems Research Group and Centre for Complex Systems, Faculty of Engineering \& IT, The University of Sydney,
    NSW 2006, Australia.}
  \thanks{C.~Finn is also with CSIRO Data61, Marsfield NSW 2122, Australia}
  \thanks{Manuscript received \today}} 


\maketitle

\begin{abstract}
  This paper examines how an event from one random variable provides pointwise mutual information about an event from another variable via probability mass exclusions.  We start by introducing probability mass diagrams, which provide a visual representation of how a prior distribution is transformed to a posterior distribution through exclusions.  With the aid of these diagrams, we identify two distinct types of probability mass exclusions---namely informative and misinformative exclusions.  Then, motivated by Fano's derivation of the pointwise mutual information, we propose four postulates which aim to decompose the pointwise mutual information into two separate informational components:\ a non-negative term associated with the informative exclusion and a non-positive term associated with the misinformative exclusions.  This yields a novel derivation of a familiar decomposition of the pointwise mutual information into entropic components.  We conclude by discussing the relevance of considering information in terms of probability mass exclusions to the ongoing effort to decompose multivariate information.
\end{abstract}


\IEEEpeerreviewmaketitle

\section{Introduction}

\IEEEPARstart{C}{onsider} three random variables $X$, $Y$, $Z$ with finite discrete state spaces
$\mathcal{X}$, $\mathcal{Y}$, $\mathcal{Z}$, and let $x$, $y$, $z$ represent events that have
occurred simultaneously in each space.  Although underappreciated in the current reference texts on
information theory \cite{cover2012,mackay2003}, both the entropy and mutual information can be
derived from first principles as fundamentally \emph{pointwise} quantities that measure the
information content of individual events rather than entire variables. The pointwise entropy
$h(x)=-\log_b(x)$, also known as the Shannon information content, quantifies the information content
of a single event $x$, while the pointwise mutual information
\begin{equation}
  i(x;y) = \log_b \frac{p(y|x)}{p(y)}
  = \log_b \frac{p(x,y)}{p(x)p(y)}
  = \log_b \frac{p(x|y)}{p(x)},
\end{equation}
quantifies the information provided by $x$ about $y$, or vice versa.\footnote{The prefix
  \emph{pointwise} has only recently become typical; both \cite{woodward1953} and \cite{fano1961}
  both referred to the pointwise mutual information as the \emph{mutual information} and then
  explicitly prefix the \emph{average} mutual information.}  To our knowledge, the first explicit
reference to pointwise information is due to Woodward and Davies~\cite{woodward1953, woodward1952}
who noted that average form of Shannon's entropy ``tempts one to enquire into other simpler methods
of derivation [of the pointwise entropy]''~\cite[p.~51]{woodward1953}.  Indeed, using two axioms
regarding the addition of information, they derived the pointwise mutual
information~\cite{woodward1952}.  Fano further formalised this idea by deriving the quantities from
four postulates that ``should be satisfied by a useful measure of information''~\cite[p.\
31]{fano1961}.

Similar to the average entropy, the pointwise entropy is non-negative.  On the contrary, unlike the
average mutual information, the pointwise mutual information is a signed measure.  A positive value
corresponds to the event $y$ raising the posterior $p(x|y)$ relative to the prior $p(x)$ so that
when $x$ occurs one would say that $y$ was informative about $x$.  In contrast, a negative value
corresponds to the event $y$ lowering the posterior $p(x|y)$ relative to the prior $p(x)$, hence
when $x$ occurs one would say that $y$ was misinformative about $x$.  Nonetheless, this
misinformation is a purely pointwise phenomena since (as observed by both Woodward and Fano) the
average information provided by the event $y$ about the variable $X$ is non-negative,
$ I(X,y) = \big\langle i(x;y) \big\rangle_{x \in \mc{X}} \geq 0$.  It follows trivially, that the
(dual average) mutual information is non-negative,
$I(X;Y) = \big\langle i(x;y) \big\rangle_{x \in \mc{X}, \, y \in \mc{Y}} \geq 0$.

\section{Information and Probability Mass Exclusions}
\label{sec:prob_mass_exc}

By definition, the pointwise information provided by $y$ about $x$ is associated with a change from
the prior $p(x)$ to the posterior $p(x|y)$. Ultimately, this change is a consequence of the
\emph{exclusion} of probability mass in the distribution $P(X)$ induced by the occurrence of the
event $y$ and inferred via the joint distribution $P(X,Y)$.  To be specific, when the event $y$
occurs, one knows that the complementary event \mbox{$\ob{y} = \{\mc{Y} \sm y\}$} did not occur;
hence, one can \emph{exclude} the probability mass in the joint distribution $P(X,Y)$ associated
with this complementary event, i.e.\ exclude $P(X,\ob{y})$.  This exclusion leaves only the
probability mass $P(X,y)$ remaining, which can be normalised to obtain the conditional distribution
$P(X|y)$.  A visual representation of how the event $y$ excludes probability mass in $P(X)$ can be
seen in the \emph{probability mass diagram} in \fig{pmass_intro}.

Since the event $x$ has also occurred, the excluded probability mass $P(X,\ob{y})$ can be divided
into two distinct categories:\ the \emph{informative exclusion} $p(\ob{x},\ob{y})$ is the portion of
the exclusion associated with the complementary event $\ob{x}$, while the \emph{misinformative
  exclusion} $p(x,\ob{y})$ is the portion of the exclusion associated with the event $x$.  The
choice of appellations is justified by considering the subsequent two special cases.  The first
special case is a \emph{purely informative exclusion} which, as depicted in \fig{pmass_cases},
occurs when the event $y$ induces exclusions which are confined to the probability mass associated
with the complementary event $\ob{x}$.  Formally, the informative exclusion ${p(\ob{x},\ob{y})}$ is
non-zero while there is no misinformative exclusion as ${p(x,\ob{y})}=0$.  Thus,
${p(x) = p(x,y) + p(x,\ob{y}) = p(x,y)}$, and hence the pointwise mutual information,
\begin{equation}
  \label{eq:pure_inf}
  i(x;y) = \log_b \frac{p(x,y)}{p(x)p(y)} = - \log_b \big( 1-p(\ob{x},\ob{y}) \big)
\end{equation}
is a strictly positive, monotonically increasing function of the size of the informative exclusion
$p(\ob{x},\ob{y})$ for fixed~$p(x)$.

\begin{figure}[t]
  \centering
  \begin{tikzpicture}

  \def\originy{0} 
  \def\height{2.8}
  \def\width{1.2}
  \def\overdrawn{0.2}
  \def\xgap{3}
  
  
  \def\originx{0}
  
  \draw[black, line width=1pt] (\originx,\originy)
    -- (\originx,\originy+\height);
  \draw[black, line width=1pt] (\originx-\overdrawn,\originy+\height)
    -- (\originx+\width+\overdrawn,\originy+\height);
  \draw[black, line width=1pt] (\originx+\width,\originy+\height)
    -- (\originx+\width,\originy);
  \draw[black, line width=1pt] (\originx+\width+\overdrawn,\originy)
    -- (\originx-\overdrawn,\originy);

  \draw[black, line width=0.75pt] (\originx,\originy+7/8*\height)
    -- (\originx+\width+\overdrawn,\originy+7/8*\height);
  \draw[black, line width=1pt] (\originx-\overdrawn,\originy+\height/2)
    -- (\originx+\width+\overdrawn,\originy+\height/2);
  \draw[black, line width=1pt] (\originx-\overdrawn,\originy+1/4*\height)
    -- (\originx+\width+\overdrawn,\originy+1/4*\height);

  \node at (\originx,3/4*\height)[anchor=east] {$x_{1}$};
  \node at (\originx,3/8*\height)[anchor=east] {$x_{2}$};
  \node at (\originx,1/8*\height)[anchor=east] {$x_{3}$};  
  \node at (\originx+\width,15/16*\height)[anchor=west] {$y_{1}$};
  \node at (\originx+\width,11/16*\height)[anchor=west] {$y_{3}$};
  \node at (\originx+\width,3/8*\height)[anchor=west] {$y_{1}$};
  \node at (\originx+\width,1/8*\height)[anchor=west] {$y_{2}$};
  \node at (\originx+\width/2,\originy+\height)[anchor=south] {$P(X,Y)$};
  \node at (\originx+\width/2,\originy+15/16*\height)[anchor=center] {$\scriptstyle 1/8$};
  \node at (\originx+\width/2,\originy+11/16*\height)[anchor=center] {$\scriptstyle 3/8$};
  \node at (\originx+\width/2,\originy+3/8*\height)[anchor=center] {$\scriptstyle 1/4$};
  \node at (\originx+\width/2,\originy+1/8*\height)[anchor=center] {$\scriptstyle 1/4$};

  
  \def\originxtwo{\xgap}
  
  \draw[black, line width=1pt] (\originxtwo,\originy)
    -- (\originxtwo,\originy+\height);
  \draw[black, line width=1pt] (\originxtwo-\overdrawn,\originy+\height)
    -- (\originxtwo+\width+\overdrawn,\originy+\height);
  \draw[black, line width=1pt] (\originxtwo+\width,\originy+\height)
    -- (\originxtwo+\width,\originy);
  \draw[black, line width=1pt] (\originxtwo+\width+\overdrawn,\originy)
    -- (\originxtwo-\overdrawn,\originy);

  \draw[black, line width=0.75pt] (\originxtwo,\originy+7/8*\height)
    -- (\originxtwo+\width+\overdrawn,\originy+7/8*\height);
  \draw[black, line width=1pt] (\originxtwo-\overdrawn,\originy+\height/2)
    -- (\originxtwo+\width+\overdrawn,\originy+\height/2);
  \draw[black, line width=1pt] (\originxtwo-\overdrawn,\originy+1/4*\height)
    -- (\originxtwo+\width+\overdrawn,\originy+1/4*\height);

  \draw[pattern=north east lines, pattern color=gray] (\originxtwo,\originy+\height/2) rectangle
    (\originxtwo+\width,\originy+7/8*\height);
  \draw[pattern=vertical lines, pattern color=gray] (\originxtwo,\originy) rectangle
    (\originxtwo+\width,\originy+\height/4);

  \node at (\originxtwo,3/4*\height)[anchor=east] {$x_{1}$};
  \node at (\originxtwo,3/8*\height)[anchor=east] {$\ob{x_1}\!$};
  \node at (\originxtwo,1/8*\height)[anchor=east] {$\ob{x_1}\!$};  
  \node at (\originxtwo+\width,15/16*\height)[anchor=west] {$y_{1}$};
  \node at (\originxtwo+\width,11/16*\height)[anchor=west] {$\ob{y_{1}}$};
  \node at (\originxtwo+\width,3/8*\height)[anchor=west] {$y_{1}$};
  \node at (\originxtwo+\width,1/8*\height)[anchor=west] {$\ob{y_{1}}$};
  \node at (\originxtwo+\width/2,\originy+\height)[anchor=south] {$P(X,y_1)$};


  \def\originxthree{2*\xgap}


%
%
%
%

    
  \draw[black, line width=1pt] (\originxthree,\originy)
    -- (\originxthree,\originy+\height);
  \draw[black, line width=1pt] (\originxthree-\overdrawn,\originy+\height)
    -- (\originxthree+\width+\overdrawn,\originy+\height);
  \draw[black, line width=1pt] (\originxthree+\width,\originy+\height)
    -- (\originxthree+\width,\originy);
  \draw[black, line width=1pt] (\originxthree+\width+\overdrawn,\originy)
    -- (\originxthree-\overdrawn,\originy);

  \draw[black, line width=1pt] (\originxthree-\overdrawn,\originy+2/3*\height)
    -- (\originxthree+\width+\overdrawn,\originy+2/3*\height);

  \node at (\originxthree,5/6*\height)[anchor=east] {$x_{1}$};
  \node at (\originxthree,2/6*\height)[anchor=east] {$\ob{x_1}\!$};
  \node at (\originxthree+\width,5/6*\height)[anchor=west] {$y_{1}$}; 
  \node at (\originxthree+\width,2/6*\height)[anchor=west] {$y_{1}$};
  \node at (\originxthree+\width/2,\originy+\height)[anchor=south] {$P(X|y_{1})$};
  \node at (\originxthree+\width/2,\originy+5/6*\height)[anchor=center] {$\scriptstyle 1/3$};
  \node at (\originxthree+\width/2,\originy+2/6*\height)[anchor=center] {$\scriptstyle 2/3$};

  
  \node at (\originxtwo-\xgap/2+\width/2 ,1/2*\height)[anchor=center] {$\implies$};
  \node at (\originxthree-\xgap/2+\width/2, 1/2*\height)[anchor=center] {$\implies$};
  
\end{tikzpicture}
  \caption[]{In probability mass diagrams, height represents the probability mass of each joint
    event from $\mc{X} \by \mc{Y}$. \emph{Left}:~the full joint distribution $P(X,Y)$.
    \emph{Centre}:~The occurrence of the event $y_1$ leads to exclusion of the probability mass
    associated with ${\ob{y_1}= \{ y_2, y_3 \}}$.  Since the event $x_1$ occurred, there is an
    informative exclusion ${p(\ob{x_1},\ob{y_1})}$ and a misinformative exclusion ${p(x_1,\ob{y_1})}$,
    represented, by convention, with vertical and diagonal hatching respectively.
    \emph{Right}:~normalising the remaining probability mass yields~$P(X|y_1)$.}
  \label{fig:pmass_intro}
\end{figure}
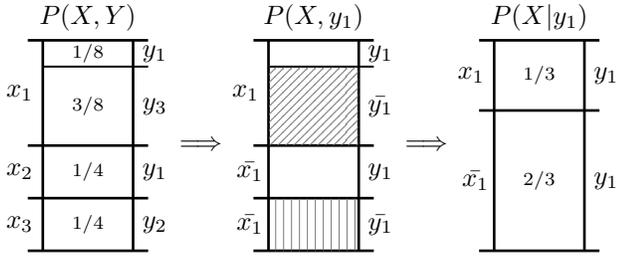

The second special case is a \emph{purely misinformative exclusion} which, as depicted in
\fig{pmass_cases}, occurs when the event $y$ induces exclusions which are confined to the
probability mass associated with the event $x$.  Formally, there is no informative exclusion as
${p(\ob{x},\ob{y})=0}$ while the misinformative exclusion ${p(x,\ob{y})}$ is non-zero. Thus,
${p(y) = 1 - p(\ob{y}) \,=\, 1 - p(x,\ob{y}) - p(\ob{x},\ob{y}) \,= 1 - p(x,\ob{y})}$, and hence,
together with ${p(x,y) = p(x) - p(x,\ob{y})}$, the pointwise mutual information,
\begin{equation}
  \label{eq:pure_mis}
  i(x;y) = \log_b \frac{p(x,y)}{p(x)p(y)}
    = \log_b \frac{1-p(x,\ob{y})/p(x)}{1-p(x,\ob{y})}
\end{equation}
is a strictly negative, monotonically decreasing function of the size of the misinformative
exclusion $p(x,\ob{y})$ for fixed~$p(x)$.






\begin{figure}[t]
  \centering
  \input{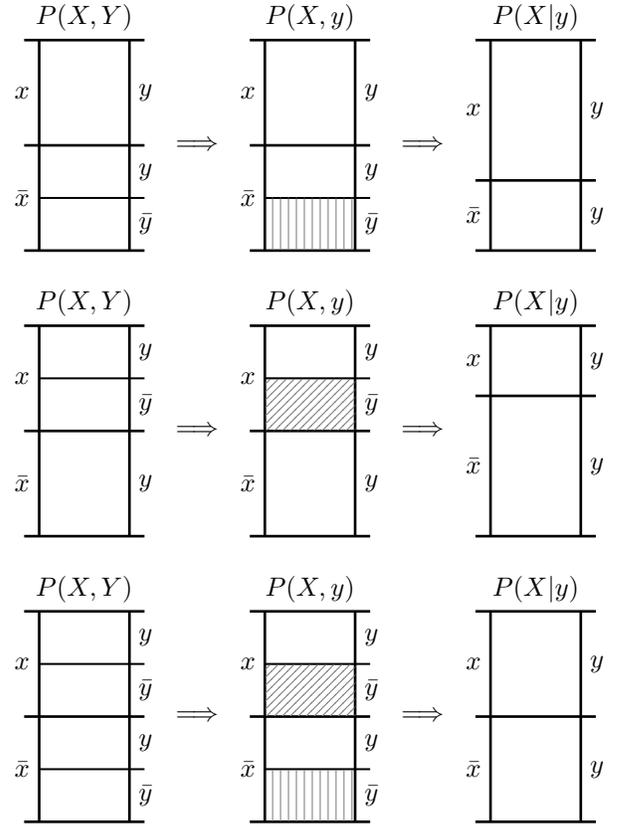}
  \caption{\emph{Top}:\ A purely informative probability mass exclusion, ${p(\ob{x},\ob{y})>0}$ and
    ${p(x,\ob{y})=0}$, leading to ${p(x|y)>p(x)}$ and hence $i(x;y)>0$. \emph{Middle}:\ A purely
    misinformative probability mass exclusion, ${p(\ob{x},\ob{y})=0}$ and ${p(x,\ob{y})>0}$, leading
    to ${p(x|y)<p(x)}$ and hence $i(x;y)<0$.  \emph{Bottom}:\ The general case $p(\ob{x},\ob{y}>0)$
    and ${p(x,\ob{y})>0}$.  Whether $i(x;y)$ turns out to be positive or negative depends on the
    balance of the exclusions.}
  \label{fig:pmass_cases}
\end{figure}

Now consider the general case depicted in \fig{pmass_cases}, where both informative and
misinformative exclusions are present simultaneously.  Given that the purely informative exclusion
yields positive pointwise mutual information, while the purely misinformative exclusion yields
negative pointwise mutual information, the question naturally arises---in the general case, can one
decompose the pointwise information into underlying informative and misinformative components each
associated with one type of exclusion?

Before attempting to address this question, there are two other important observations to be made
about probability mass exclusions.  The first observation is that an event can only ever induce an
informative exclusion about itself---if $x$ occurred then clearly that precludes the complementary
event $\ob{x}$ from having occurred. The second observation is that the exclusion process must
satisfy the chain rule of probability; in particular, as shown in \fig{pmass_chain}, there are three
equivalent ways to consider the exclusions induced in $P(X)$ by the events $y$ and $z$.  Firstly,
one could consider the information provided by the joint event $yz$ which excludes the probability
mass in $P(X)$ associated with the joint events $y\ob{z}$, $\ob{y}z$ and $\ob{y}\ob{z}$.  Secondly,
one could first consider the information provided by $y$ which excludes the probability mass in
$P(X)$ associated with the joint events $\ob{y}z$ and $\ob{y}\ob{z}$, and then subsequently consider
the information provided by $z$ which excludes the probability mass in $P(X|y)$ associated with the
joint event $y\ob{z}$.  Thirdly, one could first consider the information provided by $z$ which
excludes the probability mass in $P(X)$ associated with the joint events $y\ob{z}$ and
$\ob{y}\ob{z}$, and then subsequently consider the information provided by $y$ which excludes the
probability mass in $P(X|z)$ associated with the joint event $\ob{y} z$.  Regardless of the
chaining, one starts with the same $p(x)$ and $p(\ob{x})$ and finishes with the same $p(x|yz)$ and
$p(\ob{x}|yz)$.

\begin{figure*}[t]
  \centering
  \input{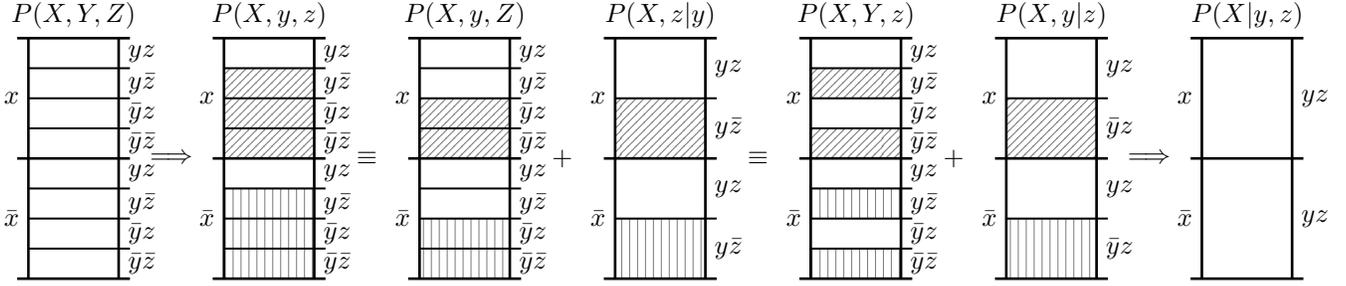}
  \caption{The probability mass exclusions must satisfy the chain rule of probability:\ there three
    equivalent ways $y$ and $z$ can provide information about $x$. 
 }
  \label{fig:pmass_chain}
\end{figure*}

Returning now to the question of decomposing the pointwise information---consider the following
postulates. Postulate~\ref{post:decomp} is a formal statement of the proposed decomposition, while
Postulate~\ref{post:monotonicity} mandates that the information associated with the exclusions
satisfies the functional relationship observed between the pointwise mutual information in both the
purely informative and purely misinformative cases.  Postulate~\ref{post:self_info} is based upon
the observation that an event can not misinform about itself, and finally,
Postulate~\ref{post:chain_rule} demands that the information associated with these exclusions must
satisfy the chain rule of probability.

\newtheorem{postulate}{Postulate}
\newcommand{\iyxp}{{i_+(y \ra x)}}
\newcommand{\iyxn}{{i_-(y \ra x)}}
\newcommand{\Iyxp}{{I_+(Y \ra X)}}
\newcommand{\Iyxn}{{I_-(Y \ra X)}}


\begin{postulate}[Decomposition]
  \label{post:decomp}
  The pointwise information provided by $y$ about $x$ can be decomposed into two non-negative
  components, such that $i(x;y) = \iyxp - \iyxn$.
\end{postulate}

\begin{postulate}[Monotonicity]
  \label{post:monotonicity}
  For all fixed $p(x,y)$ and $p(x,\ob{y})$, the function $\iyxp$ is a monotonically increasing,
  continuous function of $p(\ob{x},\ob{y})$.  For all fixed $p(\ob{x},y)$ and $p(\ob{x},\ob{y})$,
  the function $\iyxn$ is a monotonically increasing continuous function of $p(x,\ob{y})$.  For all
  fixed $p(x,y)$ and $p(\ob{x},y)$, the functions $\iyxp$ and $\iyxn$ are monotonically increasing
  and decreasing functions of $p(\ob{x},\ob{y})$, respectively.
\end{postulate}

\begin{postulate}[Self-Information]
  \label{post:self_info}
  An event cannot misinform about itself, hence $i_+(x \ra x) = i(x;x) = h(x) = - \log_b p(x)$.
\end{postulate}

\begin{postulate}[Chain Rule]
  \label{post:chain_rule}
  The functions \mbox{$\iyxp$} and \mbox{$\iyxn$} satisfy a chain rule, i.e.
  \begin{align*}
    i_+(yz \ra x) &= i_+(y \ra x) + i_+(z \ra x | y) \\
                  &= i_+(z \ra x) + i_+(y \ra x | z), \\
    i_-(yz \ra x) &= i_-(y \ra x) + i_-(z \ra x | y) \\
                  &= i_-(z \ra x) + i_-(y \ra x | z)
  \end{align*}
\end{postulate}

 
\newtheorem{theorem}{Theorem}
\newtheorem{lemma}{Lemma}
\newtheorem{corollary}{Corollary}

 
\begin{theorem}
  \label{thm:decomp}
  The unique functions satisfying the postulates are
  \begin{alignat}{4}
    \label{eq:spec}
    & \iyxp &&= h(y)   &&= - \log_b p(y),   \\
    & \iyxn &&= h(y|x) &&= - \log_b p(y|x),
    \label{eq:ambig}
  \end{alignat}
  where the base $b$ is fixed by the choice of base in Postulate~\ref{post:self_info}.
\end{theorem}

By writing these function in terms of the exclusions, it is trivial to see that \eq{spec} and
\eq{ambig} satisfy \mbox{Postulates~\ref{post:decomp}--\ref{post:chain_rule}}, i.e.
\begin{align}
  \label{eq:spec_exc}  
  \iyxp &= -\log \big(1 - p(x,\ob{y}) - p(\ob{x},\ob{y})\big), \\[5pt]
  \iyxn &= -\log \bigg(1 - \frac{p(x,\ob{y})}{p(x)} \bigg).
  \label{eq:ambig_exc}    
\end{align}
As such, the proof focuses on the uniqueness of the functions and is structured as follows:\
Lemma~\ref{lem:no_comp} considers the functional form required when $p(\ob{x})=0$, and is used in
the proof of Lemma~\ref{lem:pure_mis}; Lemmas~\ref{lem:pure_inf} and \ref{lem:pure_mis} consider the
purely informative and misinformative special cases respectively; finally, the proof of
Theorem~\ref{thm:decomp} brings these two special cases together for the general case.
 
\begin{lemma}
  \label{lem:no_comp}
  In the special case where $p(\ob{x}) = 0$, we have that $\iyxp = \iyxn = - \log_k p(y)$ where
  $k \geq b$.
\end{lemma}

\begin{proof}
  Since $p(\ob{x}) = 0$, we have that $i(x;y)=0$ and hence by Postulate~\ref{post:decomp}, that
  $\iyxp = \iyxn$.  Furthermore, we also have that $p(y) = 1 - p(x,\ob{y})$; thus, without a loss of
  generality, we will consider $\iyxn$ to be a function of $p(y)$ rather than $p(x,\ob{y})$.  As
  such, let $f(m)$ be our candidate function for $\iyxn$ where $m=\nf{1}{p(y)}$.  First consider
  choosing ${p(x,\ob{y})=0}$, such that ${m=1}$.  Postulate~\ref{post:chain_rule} demands that
  $f(1) = f(1 \cdot 1) = f(1) + f(1)$ and hence ${f(1)=0}$, i.e.\ if there is no misinformative
  exclusion, then the negative informational component should be zero.
  
  Now consider choosing $p(x,\ob{y})$ so that $m$ is a positive integer greater than $1$.  If $r$ is
  an arbitrary positive integer, then $2^r$ lies somewhere between two powers of $m$, i.e.\ there
  exists a positive integer $n$ such that
  \begin{equation}
    \label{eq:int_bound}
    m^n \leq 2^r < m^{n+1}.
  \end{equation}
  So long as the base $k$ is greater than 1, the logarithm is a monotonically increasing function,
  thus
  \begin{equation}
    \log_k m^n \leq \log_k 2^r < \log_k m^{n+1},
  \end{equation}
  or equivalently, 
  \begin{equation}
    \label{eq:log_bound}
    \frac{n}{r} \leq \frac{\log_k 2}{\log_k m} < \frac{n+1}{r}.
  \end{equation}
  By Postulate \ref{post:monotonicity}, $f(m)$ is a monotonically increasing function of $m$, hence
  applying it to \eq{int_bound} yields
  \begin{equation}
    \label{eq:f_bound}
    f(m^n) \leq f(2^r) < f(m^{n+1}). 
  \end{equation}
  Note that, by Postulate~\ref{post:chain_rule} and mathematical induction, it is trivial to
  verify that
  \begin{equation}
    \label{eq:induc_hyp}
    f(m^{n}) = n \cdot f(m).
  \end{equation}
  Hence, by \eq{f_bound} and \eq{induc_hyp}, we have that
  \begin{equation}
    \label{eq:function_bound}
    \frac{n}{r} \leq \frac{f(2)}{f(m)} < \frac{n+1}{r}.
  \end{equation}
  Now, \eq{log_bound} and \eq{function_bound} have the same bounds, hence
  \begin{equation}
    \label{eq:squeeze}
    \left| \frac{\log_k 2}{\log_k m} - \frac{f(2)}{f(m)} \right| \leq \frac{1}{r}.
  \end{equation}
  Since $m$ is fixed and $r$ is arbitrary, let \mbox{$r \ra \infty$}.  Then, by the squeeze theorem,
  we get that
  \begin{align}
    \frac{\log_k 2}{\log_k m} &= \frac{f(2)}{f(m)}, \\
  \shortintertext{and hence,}
    \label{eq:resultInteger}
    f(m) &= \log_k m.
  \end{align}

  Now consider choosing $p(x,\ob{y})$ so that $m$ is a rational number; in particular, let
  $m=\nf{s}{r}$ where $s$ and $r$ are positive integers.  By Postulate~\ref{post:chain_rule},
  \begin{equation}
    \label{eq:in_rational}
    f(s) = f(\nf{s}{r} \cdot r) = f(\nf{s}{r}) + f(r) = f(m) + f(r).
  \end{equation}
  Thus, combining \eq{resultInteger} and \eq{in_rational}, we get that
  \begin{equation}
    \label{eq:rational}
    f(m) \!=\! f(s) - f(r) = \log_k(s) - \log_k (r) = \log_k m.
  \end{equation}
  Now consider choosing $p(x,\ob{y})$ such that $m$ is a real number.  By
  Postulate~\ref{post:monotonicity}, the function \eq{rational} is the unique solution, and hence,
  $\iyxp = \iyxn = - \log_k p(y)$.

  Finally, to show that $k \geq b$, consider an event $z=y$. By Postulate~\ref{post:self_info},
  ${i_+(y \ra z)} = - \log_b p(y)$. Furthermore, since $p(\ob{z},\ob{y}) \geq p(\ob{x},\ob{y}) = 0$,
  by Postulate~\ref{post:monotonicity}, ${i_+(y \ra z)} \geq \iyxp$.  Thus,
  $-\log_b p(y) \geq -\log_k p(y)$, and hence $k \geq b$. \qedhere
  
\end{proof}


\begin{figure}[t]
  \centering
  \begin{tikzpicture}

  \def\originy{0} 
  \def\height{2.8}
  \def\width{1.2}
  \def\overdrawn{0.2}
  \def\xgap{3}


  \def\originxthree{0}

  \draw[pattern=vertical lines, pattern color=gray] (\originxthree,\originy) rectangle
    (\originxthree+\width,\originy+3/4*\height); 
  
  \draw[black, line width=1pt] (\originxthree,\originy)
    -- (\originxthree,\originy+\height);
  \draw[black, line width=1pt] (\originxthree-\overdrawn,\originy+\height)
    -- (\originxthree+\width+\overdrawn,\originy+\height);
  \draw[black, line width=1pt] (\originxthree+\width,\originy+\height)
    -- (\originxthree+\width,\originy);
  \draw[black, line width=1pt] (\originxthree+\width+\overdrawn,\originy)
    -- (\originxthree-\overdrawn,\originy);

  \draw[black, line width=0.75pt] (\originxthree,\originy+1/2*\height)
    -- (\originxthree+\width+\overdrawn,\originy+1/2*\height);
  \draw[black, line width=1pt] (\originxthree-\overdrawn,\originy+3/4*\height)
    -- (\originxthree+\width+\overdrawn,\originy+3/4*\height);
  \draw[black, line width=0.75pt] (\originxthree,\originy+1/4*\height)
    -- (\originxthree+\width+\overdrawn,\originy+1/4*\height);
   
  \node at (\originxthree,7/8*\height)[anchor=east] {$x$};
  \node at (\originxthree,3/8*\height)[anchor=east] {$\ob{x}$};
  \node at (\originxthree+\width,7/8*\height)[anchor=west] {$yz$};
  \node at (\originxthree+\width,5/8*\height)[anchor=west] {$y\ob{z}$};
  \node at (\originxthree+\width,3/8*\height)[anchor=west] {$\ob{y}z$};
  \node at (\originxthree+\width,1/8*\height)[anchor=west] {$\ob{y}\ob{z}$};
  \node at (\originxthree+\width/2,\originy+\height)[anchor=south] {$P(X,y,z)$};


  \def\originx{\xgap}
  
  \draw[pattern=vertical lines, pattern color=gray] (\originx,\originy) rectangle
    (\originx+\width,\originy+1/2*\height); 
  
  \draw[black, line width=1pt] (\originx,\originy)
    -- (\originx,\originy+\height);
  \draw[black, line width=1pt] (\originx-\overdrawn,\originy+\height)
    -- (\originx+\width+\overdrawn,\originy+\height);
  \draw[black, line width=1pt] (\originx+\width,\originy+\height)
    -- (\originx+\width,\originy);
  \draw[black, line width=1pt] (\originx+\width+\overdrawn,\originy)
    -- (\originx-\overdrawn,\originy);

  \draw[black, line width=0.75pt] (\originx,\originy+1/2*\height)
    -- (\originx+\width+\overdrawn,\originy+1/2*\height);
  \draw[black, line width=1pt] (\originx-\overdrawn,\originy+3/4*\height)
    -- (\originx+\width+\overdrawn,\originy+3/4*\height);
  \draw[black, line width=0.75pt] (\originx,\originy+1/4*\height)
    -- (\originx+\width+\overdrawn,\originy+1/4*\height);
   
  \node at (\originx,7/8*\height)[anchor=east] {$x$};
  \node at (\originx,3/8*\height)[anchor=east] {$\ob{x}$};
  \node at (\originx+\width,7/8*\height)[anchor=west] {$yz$};
  \node at (\originx+\width,5/8*\height)[anchor=west] {$y\ob{z}$};
  \node at (\originx+\width,3/8*\height)[anchor=west] {$\ob{y}z$};
  \node at (\originx+\width,1/8*\height)[anchor=west] {$\ob{y}\ob{z}$};
  \node at (\originx+\width/2,\originy+\height)[anchor=south] {$P(X,y,Z)$};

  
  \def\originxtwo{2*\xgap}


  \draw[pattern=vertical lines, pattern color=gray] (\originxtwo,\originy) rectangle
    (\originxtwo+\width,\originy+1/2*\height); 
  
  \draw[black, line width=1pt] (\originxtwo,\originy)
    -- (\originxtwo,\originy+\height);
  \draw[black, line width=1pt] (\originxtwo-\overdrawn,\originy+\height)
    -- (\originxtwo+\width+\overdrawn,\originy+\height);
  \draw[black, line width=1pt] (\originxtwo+\width,\originy+\height)
    -- (\originxtwo+\width,\originy);
  \draw[black, line width=1pt] (\originxtwo+\width+\overdrawn,\originy)
    -- (\originxtwo-\overdrawn,\originy);

  \draw[black, line width=1pt] (\originxtwo-\overdrawn,\originy+\height/2)
    -- (\originxtwo+\width+\overdrawn,\originy+\height/2);

  \node at (\originxtwo,3/4*\height)[anchor=east] {$x$};
  \node at (\originxtwo,1/4*\height)[anchor=east] {$\ob{x}$};
  \node at (\originxtwo+\width,3/4*\height)[anchor=west] {$yz$};
  \node at (\originxtwo+\width,1/4*\height)[anchor=west] {$y\ob{z}$};
  \node at (\originxtwo+\width/2,\originy+\height)[anchor=south] {$P(X,z|y)$};

  
  \node at (\originx-\xgap/2+\width/2 ,1/2*\height)[anchor=center] {$\equiv$};
  \node at (\originxtwo-\xgap/2+\width/2, 1/2*\height)[anchor=center] {$+$};

\end{tikzpicture}
  \caption{The probability mass diagram associated with \eq{pmass_lem_inf}. Lemma~\ref{lem:pure_inf}
    uses Postulates~\ref{post:self_info} and \ref{post:chain_rule} to provide a solution for the
    purely informative case. }
  \label{fig:pmass_lem_inf}
\end{figure}
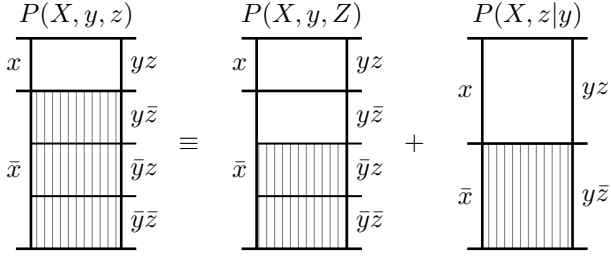

\begin{lemma}
  \label{lem:pure_inf}
  In the purely informative case where $p(x, \ob{y}) = 0$, we have that $\iyxp = - \log_b p(y)$ and
  $\iyxn = 0$.
\end{lemma}

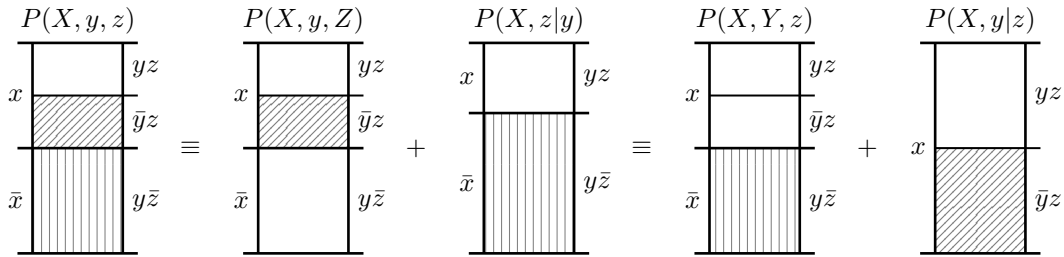
\begin{figure*}[b]
    \centering
    \begin{tikzpicture}

  \def\originy{0} 
  \def\height{2.8}
  \def\width{1.2}
  \def\overdrawn{0.2}
  \def\xgap{3}
  

  \def\originxthree{0}

  \draw[pattern=north east lines, pattern color=gray] (\originxthree,\originy+1/2*\height) rectangle
    (\originxthree+\width,\originy+3/4*\height); 
  \draw[pattern=vertical lines, pattern color=gray] (\originxthree,\originy) rectangle
    (\originxthree+\width,\originy+1/2*\height); 
  
  \draw[black, line width=1pt] (\originxthree,\originy)
    -- (\originxthree,\originy+\height);
  \draw[black, line width=1pt] (\originxthree-\overdrawn,\originy+\height)
    -- (\originxthree+\width+\overdrawn,\originy+\height);
  \draw[black, line width=1pt] (\originxthree+\width,\originy+\height)
    -- (\originxthree+\width,\originy);
  \draw[black, line width=1pt] (\originxthree+\width+\overdrawn,\originy)
    -- (\originxthree-\overdrawn,\originy);

  \draw[black, line width=0.75pt] (\originxthree,\originy+3/4*\height)
    -- (\originxthree+\width+\overdrawn,\originy+3/4*\height);
  \draw[black, line width=1pt] (\originxthree-\overdrawn,\originy+1/2*\height)
    -- (\originxthree+\width+\overdrawn,\originy+1/2*\height);
   
  \node at (\originxthree,6/8*\height)[anchor=east] {$x$};
  \node at (\originxthree,2/8*\height)[anchor=east] {$\ob{x}$};
  \node at (\originxthree+\width,7/8*\height)[anchor=west] {$yz$};
  \node at (\originxthree+\width,5/8*\height)[anchor=west] {$\ob{y}z$};
  \node at (\originxthree+\width,2/8*\height)[anchor=west] {$y\ob{z}$};
  \node at (\originxthree+\width/2,\originy+\height)[anchor=south] {$P(X,y,z)$};


  \def\originx{\xgap}
  
  \draw[pattern=north east lines, pattern color=gray] (\originx,\originy+1/2*\height) rectangle
    (\originx+\width,\originy+3/4*\height); 
  
  \draw[black, line width=1pt] (\originx,\originy)
    -- (\originx,\originy+\height);
  \draw[black, line width=1pt] (\originx-\overdrawn,\originy+\height)
    -- (\originx+\width+\overdrawn,\originy+\height);
  \draw[black, line width=1pt] (\originx+\width,\originy+\height)
    -- (\originx+\width,\originy);
  \draw[black, line width=1pt] (\originx+\width+\overdrawn,\originy)
    -- (\originx-\overdrawn,\originy);

  \draw[black, line width=0.75pt] (\originx,\originy+3/4*\height)
    -- (\originx+\width+\overdrawn,\originy+3/4*\height);
  \draw[black, line width=1pt] (\originx-\overdrawn,\originy+1/2*\height)
    -- (\originx+\width+\overdrawn,\originy+1/2*\height);
   
  \node at (\originx,6/8*\height)[anchor=east] {$x$};
  \node at (\originx,2/8*\height)[anchor=east] {$\ob{x}$};
  \node at (\originx+\width,7/8*\height)[anchor=west] {$yz$};
  \node at (\originx+\width,5/8*\height)[anchor=west] {$\ob{y}z$};
  \node at (\originx+\width,2/8*\height)[anchor=west] {$y\ob{z}$};
  \node at (\originx+\width/2,\originy+\height)[anchor=south] {$P(X,y,Z)$};

  
  \def\originxtwo{2*\xgap}


  \draw[pattern=vertical lines, pattern color=gray] (\originxtwo,\originy) rectangle
    (\originxtwo+\width,\originy+2/3*\height); 
  
  \draw[black, line width=1pt] (\originxtwo,\originy)
    -- (\originxtwo,\originy+\height);
  \draw[black, line width=1pt] (\originxtwo-\overdrawn,\originy+\height)
    -- (\originxtwo+\width+\overdrawn,\originy+\height);
  \draw[black, line width=1pt] (\originxtwo+\width,\originy+\height)
    -- (\originxtwo+\width,\originy);
  \draw[black, line width=1pt] (\originxtwo+\width+\overdrawn,\originy)
    -- (\originxtwo-\overdrawn,\originy);

  \draw[black, line width=1pt] (\originxtwo-\overdrawn,\originy+2/3*\height)
    -- (\originxtwo+\width+\overdrawn,\originy+2/3*\height);

  \node at (\originxtwo,5/6*\height)[anchor=east] {$x$};
  \node at (\originxtwo,2/6*\height)[anchor=east] {$\ob{x}$};
  \node at (\originxtwo+\width,5/6*\height)[anchor=west] {$yz$};
  \node at (\originxtwo+\width,2/6*\height)[anchor=west] {$y\ob{z}$};
  \node at (\originxtwo+\width/2,\originy+\height)[anchor=south] {$P(X,z|y)$};


  \def\originxfour{3*\xgap}
  
  \draw[pattern=vertical lines, pattern color=gray] (\originxfour,\originy) rectangle
    (\originxfour+\width,\originy+1/2*\height); 
  
  \draw[black, line width=1pt] (\originxfour,\originy)
    -- (\originxfour,\originy+\height);
  \draw[black, line width=1pt] (\originxfour-\overdrawn,\originy+\height)
    -- (\originxfour+\width+\overdrawn,\originy+\height);
  \draw[black, line width=1pt] (\originxfour+\width,\originy+\height)
    -- (\originxfour+\width,\originy);
  \draw[black, line width=1pt] (\originxfour+\width+\overdrawn,\originy)
    -- (\originxfour-\overdrawn,\originy);

  \draw[black, line width=0.75pt] (\originxfour,\originy+3/4*\height)
    -- (\originxfour+\width+\overdrawn,\originy+3/4*\height);
  \draw[black, line width=1pt] (\originxfour-\overdrawn,\originy+1/2*\height)
    -- (\originxfour+\width+\overdrawn,\originy+1/2*\height);
   
  \node at (\originxfour,6/8*\height)[anchor=east] {$x$};
  \node at (\originxfour,2/8*\height)[anchor=east] {$\ob{x}$};
  \node at (\originxfour+\width,7/8*\height)[anchor=west] {$yz$};
  \node at (\originxfour+\width,5/8*\height)[anchor=west] {$\ob{y}z$};
  \node at (\originxfour+\width,2/8*\height)[anchor=west] {$y\ob{z}$};
  \node at (\originxfour+\width/2,\originy+\height)[anchor=south] {$P(X,Y,z)$};


  \def\originxfive{4*\xgap}
  
  \draw[pattern=north east lines, pattern color=gray] (\originxfive,\originy) rectangle
    (\originxfive+\width,\originy+1/2*\height); 
  
  \draw[black, line width=1pt] (\originxfive,\originy)
    -- (\originxfive,\originy+\height);
  \draw[black, line width=1pt] (\originxfive-\overdrawn,\originy+\height)
    -- (\originxfive+\width+\overdrawn,\originy+\height);
  \draw[black, line width=1pt] (\originxfive+\width,\originy+\height)
    -- (\originxfive+\width,\originy);
  \draw[black, line width=1pt] (\originxfive+\width+\overdrawn,\originy)
    -- (\originxfive-\overdrawn,\originy);

  \draw[black, line width=0.75pt] (\originxfive,\originy+1/2*\height)
    -- (\originxfive+\width+\overdrawn,\originy+1/2*\height);
   
  \node at (\originxfive,1/2*\height)[anchor=east] {$x$};
  \node at (\originxfive+\width,6/8*\height)[anchor=west] {$yz$};
  \node at (\originxfive+\width,2/8*\height)[anchor=west] {$\ob{y}z$};
  \node at (\originxfive+\width/2,\originy+\height)[anchor=south] {$P(X,y|z)$};

  \node at (\originx-\xgap/2+\width/2,1/2*\height)[anchor=center] {$\equiv$};
  \node at (\originxtwo-\xgap/2+\width/2,1/2*\height)[anchor=center] {$+$};
  \node at (\originxfour-\xgap/2+\width/2,1/2*\height)[anchor=center] {$\equiv$};
  \node at (\originxfive-\xgap/2+\width/2,1/2*\height)[anchor=center] {$+$};

\end{tikzpicture}
    \caption{The diagram corresponding to \eq{pmass_lem_mis_pos} and \eq{pmass_lem_mis_neg}.
      Lemma~\ref{lem:pure_mis} uses Postulate~\ref{post:chain_rule} and Lemma~\ref{lem:no_comp} to
      provide a solution for the purely misinformative case.}
    \label{fig:pmass_lem_mis}
  \end{figure*}

\begin{proof}
  Consider an event $z$ such that $x = yz$ and $\ob{x}=\{y\ob{z}, \ob{y}z, \ob{y}\ob{z}\}$.  By
  Postulate~\ref{post:chain_rule},
  \begin{equation}
    \label{eq:pmass_lem_inf}
    i_+(yz \ra x) = i_+(y \ra x) + i_+(z \ra x | y),
  \end{equation}
  as depicted in \fig{pmass_lem_inf}.  By Postulate~\ref{post:self_info}, $i_+(yz \ra x) = h(x)$ and
  $i_+(z \ra x |y) = h(x|y)$, where the latter equality follows from the equivalence of the events
  $x$ and $z$ given $y$.  Furthermore, since $p(x, \ob{y}) = 0$, we have that $p(x,y)=p(x)$, and
  hence that $p(y|x)=1$.  Thus, from \eq{pmass_lem_inf}, we have that
  \begin{align}
    i_+(y \ra x) &= h(x) - h(x|y) \nn\\
                 &= h(y) - h(y|x) \nn\\
                 &= h(y).
  \end{align}
  Finally, by Postulate~\ref{post:decomp}, $i_-(y \ra x) = 0$.
\end{proof}

\newpage

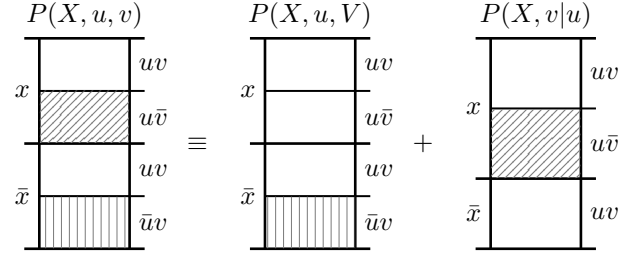
\begin{figure}[h]
  \centering
  \begin{tikzpicture}

  \def\originy{0} 
  \def\height{2.8}
  \def\width{1.2}
  \def\overdrawn{0.2}
  \def\xgap{3}
  

  \def\originxthree{0}
  
  \draw[black, line width=1pt] (\originxthree,\originy)
    -- (\originxthree,\originy+\height);
  \draw[black, line width=1pt] (\originxthree-\overdrawn,\originy+\height)
    -- (\originxthree+\width+\overdrawn,\originy+\height);
  \draw[black, line width=1pt] (\originxthree+\width,\originy+\height)
    -- (\originxthree+\width,\originy);
  \draw[black, line width=1pt] (\originxthree+\width+\overdrawn,\originy)
    -- (\originxthree-\overdrawn,\originy);

  \draw[black, line width=0.75pt] (\originxthree,\originy+3/4*\height)
    -- (\originxthree+\width+\overdrawn,\originy+3/4*\height);
  \draw[black, line width=1pt] (\originxthree-\overdrawn,\originy+\height/2)
    -- (\originxthree+\width+\overdrawn,\originy+\height/2);
  \draw[black, line width=0.75pt] (\originxthree,\originy+1/4*\height)
    -- (\originxthree+\width+\overdrawn,\originy+1/4*\height);a

  \draw[pattern=north east lines, pattern color=gray] (\originxthree,\originy+\height/2) rectangle
    (\originxthree+\width,\originy+3/4*\height);
  \draw[pattern=vertical lines, pattern color=gray] (\originxthree,\originy) rectangle
    (\originxthree+\width,\originy+\height/4);

  \node at (\originxthree,3/4*\height)[anchor=east] {$x$};
  \node at (\originxthree,1/4*\height)[anchor=east] {$\ob{x}$};
    
  \node at (\originxthree+\width,7/8*\height)[anchor=west] {$uv$};
  \node at (\originxthree+\width,5/8*\height)[anchor=west] {$u\ob{v}$};
  \node at (\originxthree+\width,3/8*\height)[anchor=west] {$uv$};
  \node at (\originxthree+\width,1/8*\height)[anchor=west] {$\ob{u}v$};
  \node at (\originxthree+\width/2,\originy+\height)[anchor=south] {$P(X,u,v)$};


  \def\originx{\xgap}
  
  \draw[black, line width=1pt] (\originx,\originy)
    -- (\originx,\originy+\height);
  \draw[black, line width=1pt] (\originx-\overdrawn,\originy+\height)
    -- (\originx+\width+\overdrawn,\originy+\height);
  \draw[black, line width=1pt] (\originx+\width,\originy+\height)
    -- (\originx+\width,\originy);
  \draw[black, line width=1pt] (\originx+\width+\overdrawn,\originy)
    -- (\originx-\overdrawn,\originy);

  \draw[black, line width=0.75pt] (\originx,\originy+3/4*\height)
    -- (\originx+\width+\overdrawn,\originy+3/4*\height);
  \draw[black, line width=1pt] (\originx-\overdrawn,\originy+\height/2)
    -- (\originx+\width+\overdrawn,\originy+\height/2);
  \draw[black, line width=0.75pt] (\originx,\originy+1/4*\height)
    -- (\originx+\width+\overdrawn,\originy+1/4*\height);
  
  \draw[pattern=vertical lines, pattern color=gray] (\originx,\originy) rectangle
    (\originx+\width,\originy+\height/4);

  \node at (\originx,3/4*\height)[anchor=east] {$x$};
  \node at (\originx,1/4*\height)[anchor=east] {$\ob{x}$};
    
  \node at (\originx+\width,7/8*\height)[anchor=west] {$uv$};
  \node at (\originx+\width,5/8*\height)[anchor=west] {$u\ob{v}$};
  \node at (\originx+\width,3/8*\height)[anchor=west] {$uv$};
  \node at (\originx+\width,1/8*\height)[anchor=west] {$\ob{u}v$};
  \node at (\originx+\width/2,\originy+\height)[anchor=south] {$P(X,u,V)$};


  \def\originxtwo{2*\xgap}
  

  %
  %
  %
  %

  \draw[black, line width=1pt] (\originxtwo,\originy)
    -- (\originxtwo,\originy+\height);
  \draw[black, line width=1pt] (\originxtwo-\overdrawn,\originy+\height)
    -- (\originxtwo+\width+\overdrawn,\originy+\height);
  \draw[black, line width=1pt] (\originxtwo+\width,\originy+\height)
    -- (\originxtwo+\width,\originy);
  \draw[black, line width=1pt] (\originxtwo+\width+\overdrawn,\originy)
    -- (\originxtwo-\overdrawn,\originy);

  \draw[black, line width=0.75pt] (\originxtwo,\originy+2/3*\height)
    -- (\originxtwo+\width+\overdrawn,\originy+2/3*\height);
  \draw[black, line width=1pt] (\originxtwo-\overdrawn,\originy+1/3*\height)
    -- (\originxtwo+\width+\overdrawn,\originy+1/3*\height);

  \draw[pattern=north east lines, pattern color=gray] (\originxtwo,\originy+1/3*\height) rectangle
    (\originxtwo+\width,\originy+2/3*\height);

  \node at (\originxtwo,4/6*\height)[anchor=east] {$x$};
  \node at (\originxtwo,1/6*\height)[anchor=east] {$\ob{x}$};
    
  \node at (\originxtwo+\width,5/6*\height)[anchor=west] {$uv$};
  \node at (\originxtwo+\width,1/2*\height)[anchor=west] {$u\ob{v}$};
  \node at (\originxtwo+\width,1/6*\height)[anchor=west] {$uv$};
  \node at (\originxtwo+\width/2,\originy+\height)[anchor=south] {$P(X,v|u)$};

  
  \node at (\originxtwo-\xgap/2+\width/2 ,1/2*\height)[anchor=center] {$+$};
  \node at (\originx-\xgap/2+\width/2, 1/2*\height)[anchor=center] {$\equiv$};

\end{tikzpicture}
  \caption{The probability mass diagram associated with \eq{pmass_thm_pos} and \eq{pmass_thm_neg}.
    Theorem~\ref{thm:decomp} uses Lemmas~\ref{lem:pure_inf} and \ref{lem:pure_mis} to provide a
    solution to the general case.}
  \label{fig:pmass_thm}
\end{figure}

\begin{lemma}
  \label{lem:pure_mis}
  In the purely misinformative case where $p(\ob{x}, \ob{y}) = 0$, we have that
  $i_+(y \ra x) = h(y) - h(y|x) - \log_k p(y|x)$ and $i_-(y \ra x) = - \log_k p(y|x)$, where
  $k \geq b$.
\end{lemma}

\begin{proof}
  Consider an event $z=x$. By Postulate~\ref{post:chain_rule},
  \begin{align}
    i_+(yz \ra x) &= i_+(y \ra x) + i_+(z \ra x | y)  \nn\\
                  &= i_+(z \ra x) + i_+(y \ra x | z), \label{eq:pmass_lem_mis_pos}\\
    i_-(yz \ra x) &= i_-(y \ra x) + i_-(z \ra x | y)  \nn\\
                  &= i_-(z \ra x) + i_-(y \ra x | z),
    \label{eq:pmass_lem_mis_neg}
  \end{align}
  as depicted in \fig{pmass_lem_mis}. Since $z=x$, by Postulate~\ref{post:self_info},
  ${i_+(z \ra x)} = h(x)$, ${i_-(z \ra x)} = 0$, ${i_+(z \ra x | y)} = h(x|y)$ and
  ${i_-(z \ra x | y)} = 0$.  Furthermore, since $p(\ob{x}|z) = 0$, by Lemma~\ref{lem:no_comp},
  ${i_+(y \ra x | z)} = {i_-(y \ra x | z)} = -\log_k p(y|z) = -\log_k p(y|x)$, hence, from
  \eq{pmass_lem_mis_pos} and \eq{pmass_lem_mis_neg}, we get that
  \begin{align}
    i_+(y \ra x) &= h(x) - h(x|y) - \log_k p(y|x)  \nn\\
                 &= h(y) - h(y|x) - \log_k p(y|x),    \\
    i_-(y \ra x) &= - \log_k p(y|x),
  \end{align}
  as required.
\end{proof}

\begin{proof}[Proof of Theorem~\ref{thm:decomp}]
  In the general case, both $p(\ob{x}, \ob{y})$ and $p(x, \ob{y})$ are non-zero.  Consider two
  events, $u$ and $v$, such that $y=uv$, $p(x, \ob{u}) = 0$ and $p(\ob{x},\ob{v}) = 0$.  By
  Postulate~\ref{post:chain_rule},
  \begin{align}
    \label{eq:pmass_thm_pos}
    i_+(y \ra x) = i_+(uv \ra x) &= i_+(u \ra x) + i_+(v \ra x | u),  \\
    i_-(y \ra x) = i_-(uv \ra x) &= i_-(u \ra x) + i_-(v \ra x | u),
    \label{eq:pmass_thm_neg}                               
  \end{align}
  as depicted in \fig{pmass_thm}.  Since $p(x, \ob{u}) = 0$, by Lemma~\ref{lem:pure_inf},
  ${i_+(u \ra x)} = h(u)$ and ${i_-(u \ra x)} = 0$; furthermore, we also have that $p(x) = p(x,u)$,
  and hence $p(v|xu) = p(uv|x)$. In addition, since $p(\ob{x}, \ob{v} | u) = 0$, by
  Lemma~\ref{lem:pure_mis}, we have that ${i_+(v \ra x | u)} = h(v|u) + h(v|xu) - \log_k p(v|xu)$
  and ${i_-(v \ra x | u)} = - \log_k p(v|xu)$ where $k \geq b$.  Therefore, by \eq{pmass_thm_pos}
  and \eq{pmass_thm_neg},
  \begin{align}
    i_+(y \ra x) &= h(u) + h(v|u) - h(v|xu) - \log_k p(v|xu)  \nn\\
                  &= h(y) - h(y|x) - \log_k p(y|x),              \\
    i_-(y \ra x) &= - \log_k p(v|xu)                          \nn\\
                  &= - \log_k p(y|x).                        
  \end{align}
  
  Finally, since Postulate~\ref{post:decomp} requires that $\iyxp \geq 0$, we have that
  $h(y) - h(y|x) - \log_k p(y|x) \geq 0$, or equivalently,
  \begin{equation}
    \log_b p(y) \leq \left( 1 - \dfrac{1}{\log_b k} \right) \log_b p(y|x).
  \end{equation}
  This must hold for all $p(y)$ and $p(y|x)$, which is only true in general for $b \geq k$.  Hence,
  $k=b$ and therefore
  \begin{align}
    i_+(y \ra x) &= h(y) - h(y|x) - \log_b p(y|x)  \nn\\
                 &= h(y),                             \\
    i_-(y \ra x) &= - \log_b p(y|x)                \\\
                 &= h(y|x).                        \nn \qedhere
  \end{align}
\end{proof}

\begin{corollary}
  The conditional decomposition of the information provided by $y$ about $x$ given $z$ is given by
  \begin{alignat}{3}
    &i_+(y \ra x | z) &&= h(y|z)  &&= - \log_b p(y|z),  \label{eq:conditional_specificity} \\
    &i_-(y \ra x | z) &&= h(y|xz) &&= - \log_b p(y|xz). \label{eq:conditional_ambiguity}
  \end{alignat}
\end{corollary}
\begin{proof}
  Follows trivially using conditional distributions.
\end{proof}

\begin{corollary}
  The joint decomposition of the information provided by $y$ and $z$ about $x$ is given by
  \begin{alignat}{3}
    &i_+(yz \ra x) &&= h(yz)   &&= - \log_b p(yz),    \\
    &i_-(yz \ra x) &&= h(yz|x) &&= - \log_b p(yz|x).
  \end{alignat}
  The joint decomposition of the information provided by $y$ about $x$ and $z$ is given by
  \begin{alignat}{3}
    &i_+(y \ra xz) &&= h(y)    &&= - \log_b p(y),     \label{eq:joint_tar_spec}  \\
    &i_-(y \ra xz) &&= h(y|xz) &&= - \log_b p(y|xz). \label{eq:joint_tar_ambig}
  \end{alignat}
\end{corollary}
\begin{proof}
  Follows trivially using joint distributions.
\end{proof}

\begin{corollary}
  We have the following three identities, 
  \begin{alignat}{2}
    &i_+(y \ra x)     &&= i_+(y \ra z),     \label{eq:ident_spec}     \\
    &i_+(y \ra x | z) &&= i_-(y \ra z),     \label{eq:ident_spec_ambig} \\
    &i_-(y \ra x | z) &&= i_-(y \ra xz).  \label{eq:ident_ambig}
  \end{alignat}
\end{corollary}
\begin{proof}
  The identity~\eq{ident_spec} follows from \eq{spec}, while \eq{ident_spec_ambig} follows from
  \eq{ambig} and \eq{conditional_specificity}; finally, \eq{ident_ambig} follows from
  \eq{conditional_ambiguity} and \eq{joint_tar_ambig}.
\end{proof}

\begin{corollary}
  The information provided by $y$ about $x$ and $z$ satisfies the following chain rule, 
  \begin{alignat}{3}
    &i(y \ra xz) &&= i(y \ra x) + i(y \ra z | x). \label{eq:chain_rule_targets}
  \end{alignat}
\end{corollary}

\begin{proof}
  Starting from the joint decomposition \eq{joint_tar_spec} and \eq{joint_tar_ambig}.  By the
  identities \eq{ident_spec} and \eq{ident_ambig}, we get that
  \begin{align}
    i(y \ra xz)
    &= i_+(y \ra xz) - i_-(y \ra xz), \nn\\
    &= i_+(y \ra x)  - i_-(y \ra z | x), \\
    \intertext{Then, by identity \eq{ident_spec_ambig}, and recomposition, we get that}
    \begin{split}
       i(y \ra xz) &= i_+(y \ra x) - i_-(y \ra x)  \\
      &\qquad + i_-(y \ra x) - i_-(y \ra z | x),
    \end{split} \nn\\
    \begin{split}
      &= i_+(y \ra x) - i_-(y \ra x)  \\
      &\qquad + i_+(y \ra z | x) - i_-(y \ra z | x),
    \end{split}\\
    &= i(y \ra x) + i( y \ra z | x). \nn \qedhere
  \end{align}
\end{proof}

Note that, in general, it is not true that ${i_+(y \ra xz)} = {i_+(y \ra x)} + {i_+(y \ra z | x)}$,
nor is it true that ${i_-(y \ra xz)} = {i_-(y \ra x)} + {i_-(y \ra z | x)}$.  Hence, although not
unexpected, it is interesting to see how the chain rule \eq{chain_rule_targets} is satisfied---the
key observation is that the positive informational component provided by $y$ about $z$ given $x$
equals the negative informational component provided by $y$ about $z$, as per \eq{ident_spec_ambig}.

In summary, the unique forms satisfying \mbox{Postulates~\ref{post:decomp}--\ref{post:chain_rule}}
are
\begin{alignat}{3}
  &i_+(y \ra x)         &&= h(y)       &&= -\log_b p(y),     \\
  &i_+(y \ra x | z )    &&= h(y|z)     &&= -\log_b p(y|z),   \\
  &i_+(yz \ra x)        &&= h(yz)      &&= -\log_b p(yz),   \\
  &i_-(y \ra x)         &&= h(y|x)     &&= -\log_b p(y|x),   \\
  &i_-(y \ra x | z ) \; &&= h(y|xz) \; &&= -\log_b p(y|xz), \\
  &i_-(yz \ra x)        &&= h(yz|x)    &&= -\log_b p(yz|x).
\end{alignat}
That is, \mbox{Postulates~\ref{post:decomp}--\ref{post:chain_rule}} decompose the pointwise
information provided by $y$ about $x$ into
\begin{align}
  \label{eq:decomp}
  i(x; y) &= \iyxp - \iyxn   \nn\\
             &= h(y) - h(y|x).  
\end{align}

\section{Discussion}
\label{sec:discussion}

Clearly, the decomposition \eq{decomp} is a well-known result, especially with regards to the
(average) mutual information.  Nonetheless, it is non-trivial that considering the pointwise mutual
information in terms of the exclusions induced by $y$ in $P(X)$ should lead to this decomposition as
opposed to the decomposition $i(x;y) = h(x) - h(x|y)$.  Indeed, this latter form is more typically
used when considering information provided by $y$ about $x$, since it states that this information
is equal to the difference between the entropy of the prior $p(x)$ and the entropy of the posterior
$p(x|y)$.  Despite this, Postulates~\ref{post:decomp}--\ref{post:chain_rule} mandate the use of the former decomposition \eq{decomp},
rather than the latter.

Recall the motivational question from Section~\ref{sec:prob_mass_exc} which asked if was possible to
decompose the pointwise information into an informative and misinformative component, each
associated with one type of exclusion.  As can be seen from \eq{spec_exc} and \eq{ambig_exc}, the
unique functions derived from the exclusions do not quite possess this precise functional
independence---although the negative informational component only depends on the size of the
misinformative exclusion $p(x,\ob{y})$, the positive component depends on the size of both the
informative exclusion $p(\ob{x},\ob{y})$ and the misinformative exclusion $p(x,\ob{x})$.  That is,
since ${p(\ob{y}) = p(x,\ob{y}) + p(\ob{x},\ob{y})}$, the positive component $\iyxp$ depends on the
total size of the exclusions induced by $y$ and hence has no functional dependence on $x$, or indeed
$X$.  Thus, $\iyxp$ quantifies the \emph{specificity} of the event $y$:\ the less likely $y$ is to
occur, the greater the total amount of probability mass excluded and therefore the greater the
potential for $y$ to inform about $x$.  On the other hand, the negative component $\iyxn$ quantifies
the \emph{ambiguity} of $y$ given $x$:\ the less likely $y$ is to coincide with the event $x$, the
greater the misinformative probability mass exclusion and therefore the greater the potential for
$y$ to misinform about $x$.  This asymmetry in the functional dependence can be seen in the two
special cases.  Decomposing the pointwise mutual information for a purely informative exclusion
yields
\begin{align}
  \label{eq:pure_inf_decomp}
  i(x;y) &= \iyxp - \iyxn                          \nn\\
         &= -\log_b \big(1 - p(\ob{x},\ob{y})\big),
\end{align}
i.e.\ only the positive informational component is non-zero.  (Note that \eq{pure_inf} is
recovered.)  On the other hand, decomposing the pointwise mutual information for a purely
informative exclusion yields
\begin{align}
  \label{eq:pure_mis_decomp}
   i(x;y) &= \iyxp - \iyxn  \nn\\
          &= - \log_b \big(1 - p(x,\ob{y})\big) + \log_b \bigg(\!1-\frac{p(x,\ob{y})}{p(x)}\bigg),  
\end{align}
i.e.\ both the positive and negative informational components are non-zero.  (Note that
\eq{pure_mis} is recovered).  Nevertheless, despite both terms being non-zero, it is clear that
${\iyxp < \iyxn}$ and hence $i(x;y) < 0$.

Now as to why one should be interested in considering information in terms of exclusions---recently,
there has been a concerted effort to quantify the shared or redundant information contained in a set
of variables about one or more target variables.  There has been particular interest focusing around
a proposed axiomatic framework for decomposing multivariate information called the partial
information decomposition \cite{williams2010}.  (There are a substantial number of publications
following on from this paper, see \cite{finn17b} and references therein.)  However, flaws have been
identified in this approach regarding ``whether different random variables carry \emph{the same}
information or just \emph{the same amount} of information''~\cite{bertschinger2013} (see also
\cite{harder2013}).  In \cite{finn17b}, exclusions are utilised to provide an operational definition
of when the events $y$ and $z$ provide the same information about $x$.  Specifically, the
information is deemed to be the same information when the events $y$ and $z$ provide the same
probability mass exclusions in $P(X)$ with respect to the event $x$.  To motivate why this approach
is appealing, consider the situation depicted in the probability mass diagram in
\fig{pmass_spec_amb} where $i(x_1; y_1) = i(x_1; z_1) = \log_2 \nf{4}{3} \text{ bit}$, but yet
\begin{align}
  i_+(y_1 \ra x_1) &= \log_2 \nf{8}{3} \text{ bit}, & i_-(y_1 \ra x_1) &= 1 \text{ bit}, \nn\\
  i_+(z_1 \ra x_1) &= \log_2 \nf{4}{3} \text{ bit}, & i_-(z_1 \ra x_1) &= 0 \text{ bit}.
\end{align}
Although the net amount of information provided by $y$ and $z$ is the same, it is in some way
different since $y$ and $z$ are different in terms of exclusions.  However, this is not the subject
of this paper---those who are interested in the operational definition of shared information based
on redundant exclusions should see \cite{finn17b}.

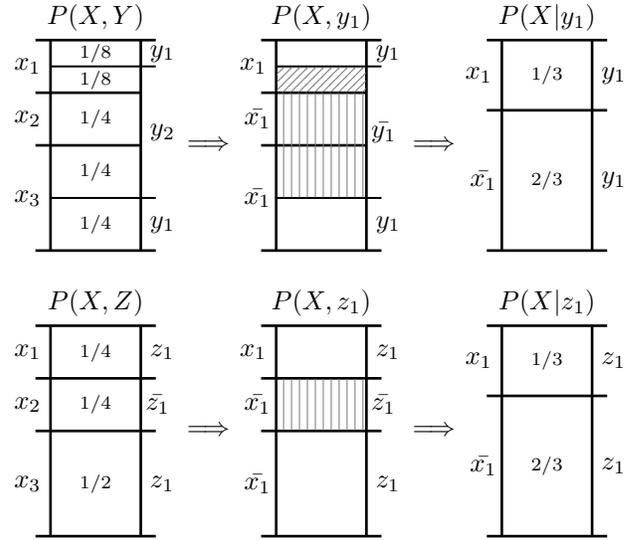
\begin{figure}[t]
  \centering
  \begin{tikzpicture}

  \def\originy{3.8} 
  \def\height{2.8}
  \def\width{1.2}
  \def\overdrawn{0.2}
  \def\xgap{3}

  
  \def\originx{0}
  
  \draw[black, line width=1pt] (\originx,\originy)
    -- (\originx,\originy+\height);
  \draw[black, line width=1pt] (\originx-\overdrawn,\originy+\height)
    -- (\originx+\width+\overdrawn,\originy+\height);
  \draw[black, line width=1pt] (\originx+\width,\originy+\height)
    -- (\originx+\width,\originy);
  \draw[black, line width=1pt] (\originx+\width+\overdrawn,\originy)
    -- (\originx-\overdrawn,\originy);

  \draw[black, line width=0.75pt] (\originx,\originy+7/8*\height)
    -- (\originx+\width+\overdrawn,\originy+7/8*\height);
  \draw[black, line width=0.75pt] (\originx,\originy+1/4*\height)
    -- (\originx+\width+\overdrawn,\originy+1/4*\height);
  
  \draw[black, line width=1pt] (\originx-\overdrawn,\originy+3/4*\height)
    -- (\originx+\width,\originy+3/4*\height);
  \draw[black, line width=1pt] (\originx-\overdrawn,\originy+\height/2)
    -- (\originx+\width,\originy+\height/2);

  \node at (\originx,\originy+7/8*\height)[anchor=east] {$x_1$};
  \node at (\originx,\originy+5/8*\height)[anchor=east] {$x_2$};
  \node at (\originx,\originy+1/4*\height)[anchor=east] {$x_3$}; 
  \node at (\originx+\width,\originy+15/16*\height)[anchor=west] {$y_1$};
  \node at (\originx+\width,\originy+9/16*\height)[anchor=west] {$y_2$};
  \node at (\originx+\width,\originy+1/8*\height)[anchor=west] {$y_1$};
  \node at (\originx+\width/2,\originy+\height)[anchor=south] {$P(X,Y)$};
  \node at (\originx+\width/2,\originy+15/16*\height)[anchor=center] {$\scriptstyle 1/8$};
  \node at (\originx+\width/2,\originy+13/16*\height)[anchor=center] {$\scriptstyle 1/8$};
  \node at (\originx+\width/2,\originy+5/8*\height)[anchor=center] {$\scriptstyle 1/4$};
  \node at (\originx+\width/2,\originy+3/8*\height)[anchor=center] {$\scriptstyle 1/4$};
  \node at (\originx+\width/2,\originy+1/8*\height)[anchor=center] {$\scriptstyle 1/4$};

  
  \def\originxtwo{\xgap}
  
  \draw[black, line width=1pt] (\originxtwo,\originy)
    -- (\originxtwo,\originy+\height);
  \draw[black, line width=1pt] (\originxtwo-\overdrawn,\originy+\height)
    -- (\originxtwo+\width+\overdrawn,\originy+\height);
  \draw[black, line width=1pt] (\originxtwo+\width,\originy+\height)
    -- (\originxtwo+\width,\originy);
  \draw[black, line width=1pt] (\originxtwo+\width+\overdrawn,\originy)
    -- (\originxtwo-\overdrawn,\originy);

  \draw[black, line width=0.75pt] (\originxtwo,\originy+7/8*\height)
    -- (\originxtwo+\width+\overdrawn,\originy+7/8*\height);
  \draw[black, line width=0.75pt] (\originxtwo,\originy+1/4*\height)
    -- (\originxtwo+\width+\overdrawn,\originy+1/4*\height);
  
  \draw[black, line width=1pt] (\originxtwo-\overdrawn,\originy+3/4*\height)
    -- (\originxtwo+\width,\originy+3/4*\height);
  \draw[black, line width=1pt] (\originxtwo-\overdrawn,\originy+\height/2)
    -- (\originxtwo+\width,\originy+\height/2);

  \draw[pattern=vertical lines, pattern color=gray] (\originxtwo,\originy+\height/4) rectangle
    (\originxtwo+\width,\originy+3/4*\height);
  \draw[pattern=north east lines, pattern color=gray] (\originxtwo,\originy+3/4*\height) rectangle
    (\originxtwo+\width,\originy+7/8*\height);

  \node at (\originxtwo,\originy+7/8*\height)[anchor=east] {$x_1$};
  \node at (\originxtwo,\originy+5/8*\height)[anchor=east] {$\ob{x_1}\!$};
  \node at (\originxtwo,\originy+1/4*\height)[anchor=east] {$\ob{x_1}\!$}; 
  \node at (\originxtwo+\width,\originy+15/16*\height)[anchor=west] {$y_1$};
  \node at (\originxtwo+\width,\originy+9/16*\height)[anchor=west] {$\!\ob{y_1}$};
  \node at (\originxtwo+\width,\originy+1/8*\height)[anchor=west] {$y_1$};
  \node at (\originxtwo+\width/2,\originy+\height)[anchor=south] {$P(X,y_1)$};


  \def\originxthree{2*\xgap}
    
  \draw[black, line width=1pt] (\originxthree,\originy)
    -- (\originxthree,\originy+\height);
  \draw[black, line width=1pt] (\originxthree-\overdrawn,\originy+\height)
    -- (\originxthree+\width+\overdrawn,\originy+\height);
  \draw[black, line width=1pt] (\originxthree+\width,\originy+\height)
    -- (\originxthree+\width,\originy);
  \draw[black, line width=1pt] (\originxthree+\width+\overdrawn,\originy)
    -- (\originxthree-\overdrawn,\originy);

  \draw[black, line width=1pt] (\originxthree-\overdrawn,\originy+2/3*\height)
    -- (\originxthree+\width+\overdrawn,\originy+2/3*\height); 

  \node at (\originxthree,\originy+5/6*\height)[anchor=east] {$x_1$};
  \node at (\originxthree,\originy+1/3*\height)[anchor=east] {$\ob{x_1}\!$};
  \node at (\originxthree+\width,\originy+5/6*\height)[anchor=west] {$y_1$};
  \node at (\originxthree+\width,\originy+1/3*\height)[anchor=west] {$y_1$}; 
  \node at (\originxthree+\width/2,\originy+\height)[anchor=south] {$P(X|y_1)$};
  \node at (\originxthree+\width/2,\originy+5/6*\height)[anchor=center] {$\scriptstyle 1/3$};
  \node at (\originxthree+\width/2,\originy+1/3*\height)[anchor=center] {$\scriptstyle 2/3$};


  \node at (\originxtwo-\xgap/2+\width/2, \originy+1/2*\height)[anchor=center] {$\implies$};
  \node at (\originxthree-\xgap/2+\width/2, \originy+1/2*\height)[anchor=center] {$\implies$};


  \def\originytwo{0} 
  
  \draw[black, line width=1pt] (\originx,\originytwo)
    -- (\originx,\originytwo+\height);
  \draw[black, line width=1pt] (\originx-\overdrawn,\originytwo+\height)
    -- (\originx+\width+\overdrawn,\originytwo+\height);
  \draw[black, line width=1pt] (\originx+\width,\originytwo+\height)
    -- (\originx+\width,\originytwo);
  \draw[black, line width=1pt] (\originx+\width+\overdrawn,\originytwo)
    -- (\originx-\overdrawn,\originytwo);

  \draw[black, line width=1pt] (\originx-\overdrawn,\originytwo+3/4*\height)
    -- (\originx+\width+\overdrawn,\originytwo+3/4*\height);
  \draw[black, line width=1pt] (\originx-\overdrawn,\originytwo+\height/2)
    -- (\originx+\width+\overdrawn,\originytwo+\height/2);

  \node at (\originx,7/8*\height)[anchor=east] {$x_1$};
  \node at (\originx,5/8*\height)[anchor=east] {$x_2$};
  \node at (\originx,1/4*\height)[anchor=east] {$x_3$}; 
  \node at (\originx+\width,7/8*\height)[anchor=west] {$z_1$};
  \node at (\originx+\width,5/8*\height)[anchor=west] {$\!\ob{z_1}$};
  \node at (\originx+\width,1/4*\height)[anchor=west] {$z_1$};
  \node at (\originx+\width/2,\originytwo+\height)[anchor=south] {$P(X,Z)$};
  \node at (\originx+\width/2,\originytwo+7/8*\height)[anchor=center] {$\scriptstyle 1/4$};
  \node at (\originx+\width/2,\originytwo+5/8*\height)[anchor=center] {$\scriptstyle 1/4$};
  \node at (\originx+\width/2,\originytwo+1/4*\height)[anchor=center] {$\scriptstyle 1/2$};

    
  \draw[black, line width=1pt] (\originxtwo,\originytwo)
    -- (\originxtwo,\originytwo+\height);
  \draw[black, line width=1pt] (\originxtwo-\overdrawn,\originytwo+\height)
    -- (\originxtwo+\width+\overdrawn,\originytwo+\height);
  \draw[black, line width=1pt] (\originxtwo+\width,\originytwo+\height)
    -- (\originxtwo+\width,\originytwo);
  \draw[black, line width=1pt] (\originxtwo+\width+\overdrawn,\originytwo)
    -- (\originxtwo-\overdrawn,\originytwo);

  \draw[black, line width=1pt] (\originxtwo-\overdrawn,\originytwo+3/4*\height)
    -- (\originxtwo+\width+\overdrawn,\originytwo+3/4*\height);
  \draw[black, line width=1pt] (\originxtwo-\overdrawn,\originytwo+\height/2)
    -- (\originxtwo+\width+\overdrawn,\originytwo+\height/2);

  \draw[pattern=vertical lines, pattern color=gray] (\originxtwo,\originytwo+\height/2) rectangle
    (\originxtwo+\width,\originytwo+3/4*\height);

  \node at (\originxtwo,7/8*\height)[anchor=east] {$x_1$};
  \node at (\originxtwo,5/8*\height)[anchor=east] {$\ob{x_1}\!$};
  \node at (\originxtwo,1/4*\height)[anchor=east] {$\ob{x_1}\!$}; 
  \node at (\originxtwo+\width,7/8*\height)[anchor=west] {$z_1$};
  \node at (\originxtwo+\width,5/8*\height)[anchor=west] {$\!\ob{z_1}$};
  \node at (\originxtwo+\width,1/4*\height)[anchor=west] {$z_1$};
  \node at (\originxtwo+\width/2,\originytwo+\height)[anchor=south] {$P(X,z_1)$};

    
  \draw[black, line width=1pt] (\originxthree,\originytwo)
    -- (\originxthree,\originytwo+\height);
  \draw[black, line width=1pt] (\originxthree-\overdrawn,\originytwo+\height)
    -- (\originxthree+\width+\overdrawn,\originytwo+\height);
  \draw[black, line width=1pt] (\originxthree+\width,\originytwo+\height)
    -- (\originxthree+\width,\originytwo);
  \draw[black, line width=1pt] (\originxthree+\width+\overdrawn,\originytwo)
    -- (\originxthree-\overdrawn,\originytwo);

  \draw[black, line width=1pt] (\originxthree-\overdrawn,\originytwo+2/3*\height)
    -- (\originxthree+\width+\overdrawn,\originytwo+2/3*\height); 

  \node at (\originxthree,5/6*\height)[anchor=east] {$x_1$};
  \node at (\originxthree,1/3*\height)[anchor=east] {$\ob{x_1}\!$};
  \node at (\originxthree+\width,5/6*\height)[anchor=west] {$z_1$};
  \node at (\originxthree+\width,1/3*\height)[anchor=west] {$z_1$}; 
  \node at (\originxthree+\width/2,\originytwo+\height)[anchor=south] {$P(X|z_1)$};
  \node at (\originxthree+\width/2,\originytwo+5/6*\height)[anchor=center] {$\scriptstyle 1/3$};
  \node at (\originxthree+\width/2,\originytwo+1/3*\height)[anchor=center] {$\scriptstyle 2/3$};
 

  \node at (\originxtwo-\xgap/2+\width/2, \originytwo+1/2*\height)[anchor=center] {$\implies$};
  \node at (\originxthree-\xgap/2+\width/2, \originytwo+1/2*\height)[anchor=center] {$\implies$};
  
\end{tikzpicture}
  \caption{\emph{Top}:~probability mass diagram for $\mc{X} \by \mc{Y}$.  \emph{Bottom}:~probability
    mass diagram for $\mc{X} \by \mc{Z}$.  Note that the events $y_1$ and $z_1$ can induce different
    exclusions in $P(X)$ and yet still yield the same conditional distributions
    ${P(X|y_1)=P(X|z_1)}$ and hence provide the same amount of information ${i(x_1;y_1)=i(x_1;z_1)}$
    about the event $x_1$.  }
  \label{fig:pmass_spec_amb}
\end{figure}

\section*{Acknowledgements}

JL was supported through the Australian Research Council DECRA grant DE160100630.  We thank Mikhail
Prokopenko, Nathan Harding, Nils Bertschinger, and Nihat Ay for helpful discussions relating to this
manuscript.  We especially thank Michael Wibral for some of our earlier discussions regarding
information and exclusions.  Finally, we would like to thank the anonymous ``Reviewer 2'' of
\cite{finn17b} for their helpful feedback regarding this paper.

\ifCLASSOPTIONcaptionsoff
  \newpage
\fi

\bibliographystyle{IEEEtran}
\bibliography{decomp_local_mi_inf_mis_contrib}

\begin{thebibliography}{1}
\providecommand{\url}[1]{#1}
\csname url@samestyle\endcsname
\providecommand{\newblock}{\relax}
\providecommand{\bibinfo}[2]{#2}
\providecommand{\BIBentrySTDinterwordspacing}{\spaceskip=0pt\relax}
\providecommand{\BIBentryALTinterwordstretchfactor}{4}
\providecommand{\BIBentryALTinterwordspacing}{\spaceskip=\fontdimen2\font plus
\BIBentryALTinterwordstretchfactor\fontdimen3\font minus
  \fontdimen4\font\relax}
\providecommand{\BIBforeignlanguage}[2]{{%
\expandafter\ifx\csname l@#1\endcsname\relax
\typeout{** WARNING: IEEEtran.bst: No hyphenation pattern has been}%
\typeout{** loaded for the language `#1'. Using the pattern for}%
\typeout{** the default language instead.}%
\else
\language=\csname l@#1\endcsname
\fi
#2}}
\providecommand{\BIBdecl}{\relax}
\BIBdecl

\bibitem{cover2012}
T.~M. Cover and J.~A. Thomas, \emph{Elements of information theory}.\hskip 1em
  plus 0.5em minus 0.4em\relax John Wiley \& Sons, 2012.

\bibitem{mackay2003}
D.~MacKay, \emph{Information Theory, Inference and Learning Algorithms}.\hskip
  1em plus 0.5em minus 0.4em\relax Cambridge University Press, 2003.

\bibitem{woodward1953}
P.~M. Woodward, \emph{Probability and information theory: with applications to
  radar}.\hskip 1em plus 0.5em minus 0.4em\relax Pergamon, 1953.

\bibitem{fano1961}
R.~Fano, \emph{Transmission of Information}.\hskip 1em plus 0.5em minus
  0.4em\relax The MIT Press, 1961.

\bibitem{woodward1952}
P.~M. Woodward and I.~L. Davies, ``Information theory and inverse probability
  in telecommunication,'' \emph{Proceedings of the IEE-Part III: Radio and
  Communication Engineering}, vol.~99, no.~58, pp. 37--44, 1952.

\bibitem{williams2010}
P.~L. Williams and R.~D. Beer, ``Nonnegative decomposition of multivariate
  information,'' \emph{arXiv:1004.2515}, 2010.

\bibitem{finn17b}
C.~Finn and J.~T. Lizier, ``Pointwise partial information decomposition using
  the specificity and ambiguity lattices,'' \emph{arXiv:1801.09010}, 2018.

\bibitem{bertschinger2013}
N.~Bertschinger, J.~Rauh, E.~Olbrich, and J.~Jost, ``Shared information---new
  insights and problems in decomposing information in complex systems,'' in
  \emph{Proceedings of the European Conference on Complex Systems 2012}.\hskip
  1em plus 0.5em minus 0.4em\relax Springer, 2013, pp. 251--269.

\bibitem{harder2013}
M.~Harder, C.~Salge, and D.~Polani, ``Bivariate measure of redundant
  information,'' \emph{Physical Review E}, vol.~87, no.~1, p. 012130, 2013.

\end{thebibliography}

%
%
%




\end{document}